\documentclass[a4paper,oneside]{amsart}

\usepackage{cite}

\renewcommand{\Tilde}{\widetilde}

\newcommand{\bs}[1]{{\boldsymbol #1}}
\newcommand{\vc}[1]{{\mathbf #1}}

\newcommand{\HH}{\mathcal{H}}
\newcommand{\CC}{\mathbb{C}}
\newcommand{\RR}{\mathbb{R}}
\newcommand{\ZZ}{\mathbb{Z}}
\newcommand{\GG}{\mathcal{G}}
\newcommand{\LL}{\mathcal{L}}

\newcommand{\TT}{\mathbb{T}}
\newcommand{\SSS}{\mathbb{S}}

\newcommand{\Arg}{\mathop{\mathrm{Arg}}}
\newcommand{\dom}{\mathop{\mathrm{dom}}}
\newcommand{\ran}{\mathop{\mathrm{ran}}}
\newcommand{\spec}{\mathop{\mathrm{spec}}}
\newcommand{\diag}{\mathop{\mathrm{diag}}}

\newtheorem{theorem}{Theorem}
\newtheorem{prop}[theorem]{Proposition}
\newtheorem{lemma}[theorem]{Lemma}

\theoremstyle{definition}

\begin{document}

\author{Konstantin Pankrashkin}

\address{Laboratoire de Math\'ematiques, Universit\'e Paris Sud,
B\^atiment 425, 91405 Orsay Cedex, France}
  
\email{konstantin.pankrashkin@math.u-psud.fr}

\title[Maryland models on quantum graphs]{Quasiperiodic surface Maryland models\\ on quantum graphs}

\begin{abstract}
We study quantum graphs corresponding to isotropic lattices
with quasiperiodic coupling constants given by the same expressions
as the coefficients of the discrete
surface Maryland model.
The absolutely continuous
and the pure point spectra are described. It is
shown that the transition between them
is governed by the Hill operator corresponding to
the edge potential.
\end{abstract}

\maketitle

\section{Introduction}

The present paper is devoted to the spectral analysis
of a special class of quasiperiodic interactions on quantum graphs.
We are going to show how some the theory of discrete
quasiperiodic operators can be transferred to the quantum graph
case using the tools of the operator theory.

The paper~\cite{GFP} provided the first explicit example of a difference
quasiperiodic operator having a pure point spectrum dense everywhere;
this operator is often referred to as the Maryland model.
Later the class of such Hamiltonians was considerably extended
in several directions, e.g. to the multidimensional case
and to more general coefficients, see e.g.~\cite{B1,FP}.
The papers~\cite{BBP,BBP1,Kh} studied interactions similar to the Maryland model
but supported by a subspace (surface Maryland model). In this case the quasiperiodic
perturbation leaves unchanged the absolutely continuous spectrum
of the unperturbed operator but produces a dense pure point spectrum
of the rest of the real line.

On the other hand, discrete operators are closely related to
the quantum graph models, i.e. differential operators
acting on geometric configuration consisting of segments,
see eg. \cite{Ku1,Ku2}. The aim of the present paper is to provide an
analog of the surface Maryland model for quantum graphs
and to study its spectral properties. The work is a natural continuation
of our previous paper \cite{KP2} where we considered full-space Maryland
quantum graph model.

\section{The model operator}\label{secmod}

Let us describe first some basic constructions for quantum graphs.
For a detailed discussion see e.g.~\cite{GS,Ku1,Ku2}. There are many
approaches to the study of the spectra of quantum graphs, we use the one from~\cite{BGP,KP1}
based on the theory of self-adjoint extensions.

We consider a quantum graph whose set of vertices is identified with $\ZZ^d$, $d\ge 2$
(i.e. we explicitly need a multidimensional lattice).
By $\vc h_j$, $j=1,\dots,d$, we denote the standard basis vectors of $\ZZ^d$.
Two vertices $\vc m$, $\vc m'$ are connected by an oriented edge $\vc m\to \vc m'$ iff $\vc m'=\vc m+\vc h_j$
for some $j\in\{1,\dots,d\}$; this edge is denoted as $(\vc m,j)$
and one says that
$\vc m$ is the initial vertex and $\vc m'\equiv \vc m+\vc h_j$ is the terminal vertex.

Replace each edge $(\vc m,j)$ by a copy of the segment $[0,1]$
in such a way that $0$ is identified with $\vc m$ and $1$ is identified with $\vc m+\vc h_j$.
In this way we arrive at a certain topological set carrying
a natural metric structure. 
The quantum state space of the system is 
\[
\HH:=\bigoplus_{(\vc m,j)\in\ZZ^d\times\{1,\dots,d\}} \HH_{\vc m,j},\quad
\HH_{\vc m,j}=\LL^2[0,1],
\]
and the vectors $f\in\HH$ will be denoted as $f=(f_{\vc m,j})$, $f_{\vc m,j}\in\HH_{\vc m,j}$,
$\vc m\in\ZZ^d$, $j=1,\dots,d$.

Let us introduce a  Schr\"odinger operator acting in $\HH$. Fix a real-valued
potential $q\in \LL^2[0,1]$ and some real constants
$\alpha(\vc m)$, $\vc m\in\ZZ^d$. Set $A:=\diag\big(\alpha(\vc m)\big)$; this is a self-adjoint operator
in $\ell^2(\ZZ^d)$. Denote by $H_A$ the  operator
acting as
\begin{subequations}
      \label{eq-sch}
\begin{equation}
        \label{eq-act}
(f_{\vc m,j})\mapsto \Big(
-f''_{\vc m,j}+q f_{\vc m,j}\Big),
\end{equation}
on functions $f=(f_{\vc m,j})\in\bigoplus_{\vc m,j} H^2[0,1]$
satisfying the following boundary conditions:
\begin{equation}
      \label{eq-cont1}
f_{\vc m,j}(0)=f_{\vc m-\vc h_k,k}(1)=:f(\vc m),\quad j,k=1,\dots,d,\quad \vc m\in\ZZ^d,
\end{equation}
(which means the continuity at all vertices)
and
\begin{equation}
f'(\vc m)=\alpha(\vc m)f(\vc m),\quad \vc m\in\ZZ^d,
\end{equation}
where
\begin{equation}
f'(\vc m):= \sum_{j=1}^d f'_{\vc m,j}(0)-\sum_{j=1}^d f'_{\vc m-h_j,j}(1).
\end{equation}
\end{subequations}
The constants $\alpha(\vc m)$ are usually referred to as \emph{Kirchhoff coupling constants}
and interpreted as the strengths of zero-range impurity potentials placed at the corresponding vertices.
The zero coupling constants correspond hence to the ideal couplings
and are usually referred to as the standard boundary conditions.

We are going to study the above operator $H_A$ for a special choice of
the coefficients $\alpha(\vc m)$. Namely, take $d_1\in\{1,\dots, d-1\}$
and set $d_2:=d-d_1$. In what follows one represents any $\vc m\in\ZZ^d$
as $\vc m=(\vc m_1,\vc m_2)$ with $\vc m_1\in\ZZ^{d_1}$ and $\vc m_2\in\ZZ^{d_2}$

Pick $g\ne 0$, $\bs\omega\in\RR^{d_2}$, $\varphi\in\RR$ with 
\begin{equation}
   \label{eq-vf1}
\varphi\ne {\bs\omega\vc m_2} \mod \dfrac12, \quad \vc m_2\in\ZZ^{d_2}
\end{equation}
and set
\[
\alpha(\vc m):=g\tan\pi \big(\bs\omega\vc m_2+\varphi\big),\quad m\in\ZZ^d.
\]
This operator will be noted simply by $H$. 

To formulate the results we need some additional constructions.
Denote by $s$ and $c$ the solutions to
$-y''+q y=zy$ satisfying $s(0;z)=c'(0;z)=0$
and $s'(0;z)=c(0;z)=1$, $z\in\CC$, and set $\eta(z):=s(1;z)+c'(1;z)$.
Consider an auxiliary one-dimensional Hill operator
\begin{equation}
   \label{eq-hill}
L=-\dfrac{d^2}{dx^2}+Q,\quad Q(x+n)=q(x), \quad (x,n)\in[0,1)\times \ZZ.
\end{equation}
It is known that $\spec L=\eta^{-1}\big([-2,2]\big)$.
\begin{theorem}
For any $\bs\omega$ and $\varphi$ one has $\spec L\subset \spec H$.
If the components of $\bs\omega$ are rationally independent, 
then the spectrum of $H$ in $\eta^{-1}\big((-2,2)\big)$
is purely absolutely continuous. If $\bs\omega$ satisfies additionally the Diophantine condition
\begin{equation}
\label{eq-vf2}
\text{there are }C,\beta>0 \text{ with }|\bs\omega \vc m_2-r|\ge C|\vc m_2|^{-\beta}
\text{ for all }\vc m_2\in\ZZ^{d_2}\setminus\{0\},\, r\in\ZZ,
\end{equation}
then the spectrum of $H$
outside $\spec L$ is dense pure point.
\end{theorem}
The above theorem is a combination of propositions \ref{spec1}, \ref{spec2}, \ref{spec3}, \ref{spec4}
whose proof will be given below.

As easily seen, the location of the absolutely continuous spectrum of $H$ is completely determined
by the spectrum of the periodic operator $L$, and (under some additional assumptions)
the rest of the spectrum is pure point. It is interesting to mention that a similar interlaced
spectrum was found recently in a completely different model involving singular potentials \cite{EF}.

\section{Some construction from the theory of self-adjoint extensions} \label{sec-ext}

In this section we recall the operator-theoretical machinery
which will be used to study the spectrum of $H$. For detailed discussion
we refer to \cite[Section 1]{BGP}.

Let $S$ be a closed linear operator in a separable
Hilbert space $\HH$ with the domain $\dom S$.
Assume that there exist
an auxiliary Hilbert space $\GG$ and two linear maps $\Gamma,\Gamma':\dom S\to \GG$ such that
\begin{itemize}
\item for any $f,g\in\dom S$ there holds
$\langle f,Sg\rangle-\langle Sf,g\rangle=\langle  \Gamma f,\Gamma' g\rangle-\langle\Gamma'f,\Gamma g\rangle$,
\item the map $(\Gamma,\Gamma'):\dom S\to\GG\oplus\GG$ is surjective,
\item the set $\ker\,(\Gamma,\Gamma')$ is dense in $\HH$.
\end{itemize}
A triple $(\GG,\Gamma, \Gamma')$ with the above properties
is called a \emph{boundary triple} for $S$.
If $S^*$ is symmetric, boundary triples deliver
an effective description of all self-adjoint restrictions of $S$.
If $A$ is a self-adjoint operator in $\GG$, then the restriction of $S$
to the vectors $f$ satisfying the abstract boundary conditions $\Gamma' f=A\Gamma f$ is a self-adjoint operator in $\HH$,
which we denote by $H_A$. Another example is the ``distinguished'' restriction $H^0$ corresponding to the boundary conditions $\Gamma f=0$.
One can show that a self-adjoint restriction $H'$ of $S$ can be represented
as $H_A$ with a suitable $A$ iff $\dom H'\cap \dom H^0= \dom S^*$; such restrictions
are called \emph{disjoint to $H^0$}. The resolvents of $H^0$ and $H_A$ as well as their spectral properties are connected by
Krein's resolvent formula, which will be described now.

Let $z\notin\spec H^0$. For $g \in \GG$ denote by $\gamma(z) g$ the unique solution to
the abstract boundary value problem
$(S-z)f=0$ with $\Gamma f=g$. Clearly, $\gamma(z)$ is a linear map from $\GG$ to $\HH$ and an isomorphism between $\GG$ and $\ker(S-z)$; it is sometimes referred to as the Krein $\gamma$-field associated with the boundary triple.
Denote also by $M(z)$ the bounded linear operator on $\GG$ given by $M(z)g=\Gamma' \gamma(z)g$; this operator
will be referred
to as the Weyl function (or the abstract Dirichlet-to-Neumann map) corresponding to the boundary triple $(\GG,\Gamma,\Gamma')$. The operator-valued functions $\gamma$ and $M$ are analytic outside $\spec H^0$, and $M(z)$ is self-adjoint for real $z$.
If these maps are known, one can relate the operators
$H_A$ and $H^0$ as follows,:
\begin{prop}\label{krein} For $z\notin \spec H^0\cup\spec H_A$ the operator $M(z)-A$ acting on $\GG$ has a
bounded inverse defined everywhere, and 
\begin{equation}
          \label{eq-krein}
(H_A-z)^{-1}=(H^0-z)^{-1}-\gamma(z)\big(M(z)-A\big)^{-1}\gamma(\bar z)^*.
\end{equation}
In particular, the set $\spec H_A\setminus\spec H^0$ consists exactly of
$z\in\RR\setminus\spec H^0$ such that $0\in\spec\big(M(z)-A\big)$.
The same correspondence holds for the eigenvalues, i.e. $z\in \RR\setminus\spec H^0$
is an eigenvalue of $H_A$ iff $0$ is an eigenvalue of $M(z)-A$,
and $\gamma(z)$ is an isomorphism of the corresponding eigensubspaces.
\end{prop}

The maps $\gamma$ and $M$ satisfy a number of important properties.
In particular, $\gamma$ and $M$ depend analytically on their argument (outside of $\spec H^0$),
$M(z)$ satisfies $M(\bar z)=M(z)^*$ and
\begin{gather}
        \label{eq-meg}
\text{for any non-real $z$ there is $c_z>0$ with}
\dfrac{\Im M(z)}{\Im z}\ge c_z, \text{ and}\\
        \label{eq-meg2}
M'(\lambda)=\gamma(\lambda)^*\gamma(\lambda)>0 \text{ for } \lambda\in\RR\setminus\spec H^0.
\end{gather}
Furthermore,
\begin{equation}\label{eq-gz}
\gamma(z)^* f=0 \text{ for any }f\in\ker(S-z)^\perp\equiv\ran \gamma(z)^\perp,
\end{equation}
see \cite[Section~1]{BGP} for more details.

Below it will be useful to have a certain relationship between the resolvent of $H_A$
and  that of the operator $H_0$ (i.e. $H_A$ with $A=0$).
Clearly, $(\GG,\Tilde \Gamma,\Tilde\Gamma')$ with $\Tilde\Gamma:=\Gamma'$ and $\Tilde\Gamma':=-\Gamma$
is a new boundary triple for $S$. With respect to this boundary triple $H_0$
is the distinguished extension, and one can easily calculate (at least for non-real $z$) the corresponding
Krein $\gamma$-field $\Tilde\gamma(z):=\gamma(z) M(z)^{-1}$ and the
Weyl function $\Tilde M(z):=-M(z)^{-1}$ which then extend by the analyticity to
$\RR\setminus\spec H_0$. The operator $H_A$ corresponds now to the boundary conditions:
$\Tilde \Gamma f =0$ for $\Gamma f\in\ker A$ and
$\Tilde \Gamma' f =-A^{-1} \Tilde\Gamma f$ otherwise.
Hence, if the operator $A$
is not invertible, one cannot use the above proposition \ref{krein} for the resolvents.
One can avoid this difficulty
as follows. Denote $\GG':=\ker A^{\perp}$; clearly, $\GG'$ is a closed linear subspace
of $\GG$. Denote by $P$ the orthogonal projection from $\GG$ to $\GG'$
and set $\Pi:=P\Tilde\Gamma P$ and $\Pi':=P\Tilde\Gamma' P$ and $S':=S|_{\Tilde\Gamma^{-1}(\GG')}$,
then $(\GG',\Pi,\Pi')$ is a boundary triple for $S'$ with the Krein field
$\nu(z):=P\Tilde \gamma(z)P$ and the Weyl function $N(z)=P\Tilde M(z)P$,
and $H_A$ and $H_0$ become disjoint self-adjoint restrictions of $S'$.
Hence, one has the following resolvent formula
\begin{equation}
      \label{eq-nz}
(H_0-z)^{-1}-(H_A-z)^{-1}=\nu(z)\big(N(z)-B\big)\nu^*(\Bar z), \quad B=-PA^{-1}P,
\end{equation}
or
\begin{multline}
(H_0-z)^{-1}-(H_A-z)^{-1}\\
=P\gamma(z) M(z)^{-1}\big(PA^{-1}P-PM(z)^{-1}P\big)M(z)^{-1}P\gamma^*(\Bar z).
\end{multline}

\section{Resolvents for the quantum graph}

We are going to show now how the constructions of section \ref{sec-ext}
apply to the quantum graph Hamiltonian $H$ (see Section \ref{secmod}).

Denote by $S$ the operator acting as \eqref{eq-act} on the functions $f$
satisfying only the boundary conditions \eqref{eq-cont1}. On the domain
of $S$ one can define linear maps
\[
f\mapsto \Gamma f:= \big(f(\vc m)\big)_{\vc m\in\ZZ^d}\in \ell^2(\ZZ^d),\quad
f\mapsto \Gamma' f:= \big(f'(\vc m)\big)_{\vc m\in\ZZ^d}\in \ell^2(\ZZ^d).
\]
One can show that $(\ZZ^d,\Gamma,\Gamma')$ form a boundary triple
for $S$. The distinguished restriction to $\ker\Gamma$,
$H^0$, acts as \eqref{eq-act}
on functions $(f_{\vc m,j})$ with $f_{\vc m,j}\in H^2[0,1]$
satisfying the Dirichlet boundary conditions,
$f_{\vc m,j}(0)=f_{\vc m,j}(1)=0$ for all $m,j$, and the spectrum
of $H^0$ is just the Dirichlet spectrum of
$-\dfrac{d^2}{dt^2}+q$ on the segments $[0,1]$;
we will refer to $\spec H^0$ as to the \emph{Dirichlet spectrum}
of the graph.

Let us construct the maps $\gamma(z)$ and $M(z)$ for the above boundary triple.
In in terms of these functions $s$ and $c$ one has obviously
\begin{multline*}
\big(\gamma(z)\xi\big)_{\vc m,j}(t)=
\dfrac{1}{s(1;z)}\Big(\xi(\vc m+\vc h_j)s(t;z)\\
+\xi(\vc m)\big(
s(1;z) c(t;z)-c(1;z)s_j(t;z)
\big)\Big),\\
t\in[0,1],\quad (\vc m,j)\in\ZZ^d\times\{1,\dots,d\}.
\end{multline*}
The corresponding Weyl function $M(z):\ell^2(\ZZ^d)\to \ell^2(\ZZ^d)$
is given by
\[ M(z)\xi(\vc m)=
\dfrac{1}{s(1;z)} \sum_{j=1}^d \cdot\Big( \xi(\vc m-\vc h_j)+\xi(\vc m+\vc h_j)
-\eta(z)\,\xi(\vc m)\Big),\quad \xi\in\ell^2(\ZZ^d),
\]
where $\eta(z):=c(1;z)+s'(1;z)$ is the Hill discriminant associated with the potential $q$.
It is useful to introduce the discrete Hamiltonian $\Delta_d$ in $\ell^2(\ZZ^d)$ by
\[
\Delta_d \xi(\vc m)=\sum_{\vc m':|\vc m- \vc m'|=1} \xi(\vc m'),
\]
then one has obviously
\begin{equation}
       \label{eq-lapl}
M(z):= a(z) \big(\Delta_d-d \eta(z)\big), \quad a(z):=\dfrac{1}{s(1;z)}.
\end{equation}
By proposition \ref{krein}, outside of the discrete set $\spec H^0$, the spectrum of $H$
consists of the real $z$ satisfying $0\in\spec \big( M(z)-A\big)$.
It is important to emphasize that, for real $z$, the operator $ M(z)-A$
is exactly the surface Maryland model studied in \cite{BBP}.
The results of \cite{BBP1} imply that
\begin{itemize}
\item the operator $M(z)-A$ has no bounded inverse if $\big|\eta(z)\big|\le 2$,
\item if the components of $\bs\omega$ are rationally independent, then the spectrum
in the interval $a(z)\big(2d-\eta(z),2d+\eta(z) \big)$ is purely absolutely
continuous,
\item for Diophantine $\bs\omega$ the rest of the real line is covered
by the dense pure point spectrum.
\end{itemize}
Using first of these properties and proposition \ref{krein} one immediately obtains
\begin{prop}\label{spec1} $\spec L\subset \spec H$.
\end{prop}
(It is sufficient to recall that the set $\big|\eta(z)\big|\le 2$ coincides with spectrum of $L$.)
Nevertheless, one is not able to conclude about the spectral nature of $H$ from that of $M(z)-A$
using just the general theory of self-adjoint extensions \cite{BrMN}.
We are going to use some additional considerations from \cite{BBP1,FP}
in order to understand completely the spectral properties of $H$.

We will also use the formula \eqref{eq-nz} relating the resolvent of $H_A$
and $H_0$. Denoting by $P$ the orthogonal projection from $\ell^2(\ZZ^d)$
to $\ell^2(\ZZ^{d_2})$ on obtains $N(z)=-PM(z)P$.
The parameter operator $B:=-(PAP)^{-1}$ is the multiplication
by $-g^{-1}\cot\pi(\bs \omega \vc m_2+\varphi)\equiv g^{-1}\tan\pi(\bs \omega \vc m_2+\varphi+1/2)$.
It is also important to emphasize that, as shown in \cite{KP1},
one has $\spec H_0=\spec L\cup\spec H^0$
where $L$ is the one-dimensional Hill operator \eqref{eq-hill},
and each point of $\spec H^0$ is an infinitely degenerate eigenvalue.

\section{The absolutely continuous spectrum}

Below we will use actively the Fourier transform.  Denote $\SSS^1:=\{z\in\CC,\,|z|=1\}$
 and $\TT^n:=\underbrace{\SSS^1\times\dots\times\SSS^1}_{n \text{ times}}\subset\CC^n$.
For $\bs \theta=(\theta_1,\dots,\theta_n)\subset\CC^n$
and $\vc p=(p_1,\dots,p_n)\in\ZZ^n$ we write $\bs \theta^\vc p:=\theta_1^{p_1}\dots\theta_n^{p_n}$,
and in this context $k\in\ZZ$ will be identified with the vector $(k,\dots,k)\in\ZZ^n$,
i.e. $\bs\theta^{-1}:=\theta_1^{-1}\dots\theta_l^{-1}$ etc. 
We denote by $F_n$ the Fourier transform carrying $\ell^2(\ZZ^n)$ to $\LL^2(\TT^n)$,
\[
F_n\psi(\bs\theta)=\sum_{\vc n\in\ZZ^n} \psi(\vc n) \bs \theta^\vc n,\quad
F_n^{-1}f(\vc n)=\dfrac{1}{(2\pi i)^n} \int_{\TT^n} f(\bs\theta) \bs\theta^{-\vc n-1}d\bs\theta.
\]
Each $\bs\theta\in\TT^d$ will be represented $\bs\theta=(\bs\theta_1,\bs\theta_2)$
with $\bs \theta_1\in\TT^{d_1}$ and $\bs\theta_2\in\TT^{d_2}$.

We will repeat first some algebraic manipulations in the spirit of~\cite{BBP}.
Without loss of generality assume $g>0$ (otherwise one can change the signs
of $\bs\omega$ and $\varphi$).

Consider the operator $L(z):=M(z)-A=M(z)+PvP$, where
$v$ is an operator in $\ell^2(\ZZ^{d_2})$ acting
as $v f(\vc m_2)=-g\tan(\bs\omega\vc m_2+\varphi)f(\vc m_2)$, $\vc m_2\in\ZZ^{d_2}$.
For $\Im z\ne 0$ the operator $M(z)$ is invertible (as its imaginary part is non-degenerate) and one has
\[
L(z)^{-1}=M(z)^{-1}-M(z)^{-1} T(z) M(z)^{-1}, \quad
T(z)=v-T(z) M(z)^{-1} v. 
\]
Obviously one can write $T(z)=P t(z)P$ where the operator $t(z)$ acting in $\ell^2(\ZZ^{d_2})$
satisfies $t(z)=v+t(z) N(z) v$.
Formally one has $t(z)=v(1-N(z) v)^{-1}$, and it is needed to show that
the operator in question is really invertible. 

Let $U$ be the unitary operator in $\ell^2(\ZZ^{d_2})$ defined by the relation
\[
(u f)(\vc m_2)=e^{-2\pi i \bs\omega \vc m_2} f(\vc m_2),
\] 
then, denoting $\chi:=e^{-2\pi i \varphi}$, one can write
\[
v=-\dfrac{g}{i} \dfrac{1-\chi U}{1+\chi U}.
\]
As $\Im N(z)\ge 0$ for $\Im z\ge 0$, the operator $i+g N(z)$
is invertible for such $z$. Hence, for $\Im z\ge 0$ 
after a simple algebra one obtains
\[
1-N(z) v = \big(g N(z)+i\big)(1-b(z)\chi U)\big( i (1+\chi U)\big)^{-1}
\]
where $b(z)=\big(g N(z)-i\big)\big(i+g N(z)\big)^{-1}$. 
In order to represent the inverse operator in terms of the Neumann series
it is sufficient to show that $|b(z)|<1$ for some $z$.
To see this, it is useful to pass to the Fourier representation.

For $\lambda\in\CC$ denote $G_d(\lambda):=(\Delta_d-\lambda)^{-1}$.
Recall that in the Fourier representation $\Delta_d$
becomes the multiplication by the function $\Delta_d(\bs \theta)=\sum_j (\theta_j+\theta_j^{-1})$, hence
the matrix of $G_d(\lambda)$ is given by
\[
G_d(\vc m-\vc m';\lambda)=\dfrac{1}{(2\pi i)^d}\int_{\TT^d} \dfrac{\theta^{-(\vc m-\vc m')-1} d\bs \theta}{\Delta_d(\bs\theta)-\lambda}.
\]
On the other hand, the matrix of the operator $N(z)$
is $N(\vc m_2-\vc m'_2;z)=-a(z)^{-1}G_d\big((0,\vc m_2)-(0,\vc m'_2);d\eta(z)\big)$, hence
\begin{multline}
       \label{eq-nzint}
N(\vc m_2-\vc m'_2;z)=-a(z)^{-1}\dfrac{1}{(2\pi i)^d}
\int_{\TT^{d_2}} \bs\theta_2^{-(\vc m_2-\vc m'_2)-1} d\vc \theta_2
\int_{\TT^{d_1}} \dfrac{\bs \theta_1^{-1} d\bs \theta_1}{\Delta_d(\bs\theta)-d\eta(z)}\\
=-a(z)^{-1}
\dfrac{1}{(2\pi i)^{d_2}}
\int_{\TT^{d_2}} G_{d_1}\big(\vc 0;d\eta(z)-\Delta_{d_2}(\bs\theta_2)\big)\bs\theta_2^{-(\vc m_2-\vc m'_2)-1} d\bs\theta_2.
\end{multline}
In particular, it is clear that in the Fourier representation $N(z)$
is the multiplication by the function
\begin{equation}
         \label{eq-nfun}
N(\bs \theta_2;z)=-a(z)^{-1}G_{d_1}\big(\vc 0;d\eta(z)-\Delta_{d_2}(\bs\theta_2)\big).
\end{equation}
As $\Im N(z)>0$ for $\Im z>0$, the imaginary part $\Im N(\bs\theta_2;z)$
is positive for such $z$. The operator $b(z)$ in the Fourier representation
becomes the multiplication by the function 
\[
b(\bs \theta_2,z)=\dfrac{g N(\bs \theta_2,z)-i}{g N(\bs\theta_2,z)+i},
\]
hence $\|b(z)\|\equiv \sup_{\bs\theta_2} |b(\bs \theta_2,z)|<1$ for $\Im z>0$.
Therefore, one can represent
\begin{multline}
t(z)=v(1-N(z) v)^{-1}\\
=-g(1-\chi U)\big(1-b(z)\chi U\big)^{-1}\big(gN(z)+i\big)^{-1}\\
=-g(1-\chi U) \sum_{m=0}^\infty \chi^m \big(b(z)U\big)^m\\
=-g\big(gN(z)+i\big)^{-1}\Big(1-
2i  \sum_{m=1}^\infty \big(gN(z)+i\big)^{-1} U \big(b(z)U\big)^{m-1}\Big).
\end{multline}
and one has
\[
\big(M(z)-A\big)^{-1}=M(z)^{-1} -M(z)^{-1} P t(z) P M(z)^{-1}.
\]
After these preparations we can prove
\begin{prop}\label{spec2}
Denote $I:=\eta^{-1}\big((-2,2)\big)$.
If the vector $\bs\omega$ has rationally independent components, then
the operator $H$ has only absolutely continuous spectrum in $I$.
\end{prop}

\begin{proof}
According to the general spectral theory we need to show that there exists a dense subset $\LL$
of $\HH$ such that the limit $\lim_{\varepsilon\to 0+}\Im\langle f, (H-\lambda-i\varepsilon)^{-1} f\rangle$
exists and is finite for all $g\in\LL$ and $\lambda\in I.$

Represent 
\begin{equation}
        \label{eq-h01}
\HH=\HH_0+\HH_1,\quad
\HH_0:=\Big(\bigcup_{\Im z \ne 0} \gamma(z)\big(\ell^2(\ZZ^d)\big)\Big)^\perp,
\quad
\HH_1:=\HH_0^\perp;
\end{equation}
in other words, $\HH_1$ is the closure of the linear hull of the set $\{\gamma(z)\varphi:\, \Im z\ne 0,\,\varphi\in\ell^2(\ZZ^d)\}$.

By the Krein resolvent formula, for any $f\in\HH_0$ and any $z$ with $\Im z\ne 0$ one has $\gamma^*(z)f=0$.
Hence, by~\eqref{eq-krein}, there holds $(H-z)^{-1}f=(H^0-z)^{-1}f$, hence
$\lim_{\varepsilon\to 0+}\Im\langle f, (H-\lambda-i\varepsilon)^{-1} f\rangle=0$n
because $(H^0-\lambda)^{-1}$ is a bounded self-adjoint operator.

Consider the vectors $f=\gamma(\zeta)h$ for $h=\big(M(\zeta)-A\big)^{-1}\xi$, $\Im\zeta\ne0$.
These vectors form a dense subset in $\HH_1$ as $\xi$ runs over a dense subset of $\ell^2(\ZZ^d)$.
By elementary calculations (see e.g. section~3 in~\cite{BGP}) one can write
\[
(H-\lambda-i\varepsilon)^{-1}f=
\dfrac{1}{\zeta-\lambda-i\varepsilon}\,\Big(
f-\gamma(\lambda+i\varepsilon) \big(M(\lambda+i\varepsilon)-A\big)^{-1}\xi\Big).
\]
Hence it is sufficient to show that $\lim_{\varepsilon\to 0+}\Im\langle \delta_{\vc m}, (M(\lambda+i\varepsilon)-A)^{-1} \delta_{\vc m}\rangle$
exists and is finite for any $\vc m\in\ZZ^d$. In view of the series representation for $(M(z)-A)^{-1}$
it is sufficient to show that the series converges for real $z\in I$ and not only
for $\Im z>0$. On the other hand, $(M(z)-A)^{-1}=a(z)^{-1}(\Delta_d-d\eta(z)-a(z)^{-1}A)^{-1}$, and
it is shown in \cite[Theorem 3.1]{BBP1} that $\lim_{\varepsilon\to 0+}\Im\langle \delta_{\vc m}, (\Delta_d-b-c A)^{-1} \delta_{\vc m}\rangle$
exists and is finite for any $b\in(-2d,2d)$ and any $c\in\RR$. This completes the proof.
\end{proof}

\section{The pure point spectrum}

In this section we will use the second version of the resolvent formula, Eq. \eqref{eq-nz}.
Hence for $z\notin\spec H_0$ we have the equivalence $z\in \spec H$ iff
$0\in \spec (N(z)-B)$. Here $N$ is a translationally invariant operator in $\ell^2(\ZZ^{d_2})$ whose matrix
elements are given by \eqref{eq-nzint}, and the operator $B$, as already mentioned below,
in the multiplication by the sequence $g^{-1}\tan\pi(\bs \omega \vc m_2+\varphi+1/2)$.
It is useful to set $g'=-g$, $\bs\omega':=-\bs\omega$, $\varphi':=-\varphi-1/2$,
then $B$ becomes a multiplication by $-g'\tan\pi(\bs \omega' \vc m_2+\varphi')$
with $g'>0$.

To alleviate the notation, below we will write $d$ instead of $d_2$ and drop the indices
for $g'$, $\bs\omega'$ and $\varphi'$ as this does not lead to confusions.

Introduce the operators
\[
D(z):=\big(N(z)-ig\big)^{-1}, \quad
C(z):=-\big(N(z)+ig\big)\big(N(z)-ig\big)^{-1};
\]
they are defined at least for $z$ with $\Re z\notin\spec H_0$ and $|\Im z|$ sufficiently small.
One can write for such $z$ the identity
\begin{equation}
     \label{eq-MA}
N(z)-B=D(z)^{-1}\big(1-\chi C(z)U\big)(1+\chi U)^{-1} .
\end{equation}
Recall that under the Fourier transform $N(z)$ becomes the multiplication by the function
$N(\bs \theta,z)$ given by \eqref{eq-nfun}, the operators $D(z)$ and $C(z)$ become the multiplications
by $D(\bs \theta,z):=\big(N(\bs\theta,z)-ig\big)^{-1}$
by $C(\bs\theta,z):=-\big(N(\bs\theta,z)+ig\big)\big(N(\vc\theta,z)-ig\big)^{-1}$,
respectively, and $U$ becomes a shift operator,
$U k(\bs\theta)=k(e^{-2\pi i \omega_1}\theta_1,\dots,e^{-2\pi i \omega_d}\theta_d)$.

Consider an arbitrary segment $[a,b]\subset\RR\setminus\spec H_0$.
Recall that the spectrum of $H_0$ coincides with the spectrum of $L$
up to the discrete set $\spec H^0$. 
Eq.~\eqref{eq-meg}, the analyticity of $\gamma$, and the self-adjointness
of $N(z)$ for real $z$ imply the existence of $\delta'>0$
such that $\|\Im N(z)\|\le g/2$ for $z\in Z:=\{z\in\CC:\, |\Im z|\le\delta',\ \Re z\in[a,b]\}$.
At the same time, this means that $|\Im N(\bs\theta,z)|\le g/2$ for $z\in Z$.
As follows from the integral representation, $N(\bs\theta,z)$ can be continued to an analytic function
in $Z\times \Theta$, $\Theta:=\{\bs\theta\subset \CC^d: r<|\theta_j|<R\}$, $0<r<1<R<\infty$.
Choosing $r$ and $R$ sufficiently close to $1$ one immediately sees
that the function
\[
C(\bs \theta,z):=\dfrac{g^2-\big(\Im N(\bs \theta,z)\big)^2-\big(\Re N(\bs \theta,z)\big)^2-2ig\Re N(\bs\theta,z)}{|N(\bs \theta,z)-ig\big|^2}
\]
does not take values in $(-\infty,0)$ for $(\bs\theta,z)\in \Theta\times Z$.
Therefore, the function $f(\bs \theta,z):=\log C(\bs\theta,z)$ is well-defined
and analytic in $\Theta\times Z$,
where $\log$ denotes the principal branch of the logarithm.
The Diophantine property~\eqref{eq-vf2} implies (see \cite[Lemma 3.2]{FP}) that
the operator $1-U$ is a bijection on the set of functions $v$ analytic in $\Theta$
with
\[
\int_{\TT^d} v(\bs \theta)\bs \theta^{-1}d\bs\theta=0.
\]
Hence the function $t(\bs \theta,z):=(1-U)^{-1}\big(f(\bs\theta,z)-f_0(z)\big)$
is well-defined and analytic in $Z\times\Theta$, where
\begin{equation}
     \label{eq-f0}
f_0(z):=
\dfrac{1}{(2\pi i)^d}\int_{\TT^d} f(\bs\theta,z)\bs\theta^{-1}d\bs\theta.
\end{equation}

\begin{lemma}\label{prop-1}
The function $f_0$ is analytic in $Z$,
\begin{gather}
   \label{eq-f01}
\Re f_0(z)<0 \quad \text{for} \quad\Im z>0,\\
   \label{eq-f02}
\Re f(\bs \theta,z)=\Re t(\bs \theta,z)=\Re f_0(z)=0 \quad \text{for} \quad \Im z=0.
\end{gather}
For real $\lambda$ one has $f_0(\lambda)=2i \sigma(\lambda)$, where
\[
\sigma(\lambda)=\dfrac{1}{(2\pi i)^d}\int_{\TT^d} \arctan \dfrac{N(\bs\theta,\lambda)}{g}\,\bs\theta^{-1}d\bs\theta.
\]
The function $\sigma$ is real-valued,
strictly increasing, and continuously differentiable on $[a,b]$.
\end{lemma}

\begin{proof}
The analyticity of $f_0$ follows from its integral representation.
Eq.~\eqref{eq-f01} follows from \eqref{eq-f0} if one takes into account
the inequalities $\Im N(\bs\theta,z)>0$ for $\Im z>0$ and $\Re\log z<0$
for $|z|<1$. Equalities~\eqref{eq-f01} follows from from \eqref{eq-f0}
and the real-valuedness of $N(\bs\theta,z)$ for real $z$.

By elementary calculations, for $x\in\RR$ and $y>0$ one  has
\begin{equation}
          \label{eq-hh2}
g_1(x):=\dfrac{1}{2i}\log \dfrac{iy+x}{iy-x}\equiv \arctan \dfrac{x}{y}=:g_2(x).
\end{equation}
In fact, this follows from 
\begin{equation}
    \label{eq-hh}
g'_1(x)=g'_2(x)= \dfrac{y}{x^2+y^2}
\end{equation}
and $g_1(0)=g_2(0)=0$.
Eq.~\eqref{eq-hh2} obviously implies $f_0(\lambda)=2i\sigma(\lambda)$ for $\lambda\in\RR$.
Furthermore, as follows from \eqref{eq-hh},
\[
\sigma'(\lambda)= \dfrac{1}{(2\pi i)^d}\int_{\TT^d} \dfrac{g N'_\lambda(\bs\theta,\lambda)}{N(\bs\theta,\lambda)^2+g^2} \,\bs\theta^{-1}d\bs\theta,
\]
and, by \eqref{eq-meg2}, $\sigma'(\lambda)>0$.
\end{proof}

An immediate corollary of the analyticity of $f_0$ and of \eqref{eq-f01} is
\begin{lemma}\label{lem-eps}
There exists $\varepsilon_0>0$ such that $\big|e^{f_0(\lambda)}\xi-1\big|\le 2\big|e^{f_0(\lambda+i\varepsilon)}\xi-1\big|$
for all $\xi\in\SSS^1$, $\lambda\in[a,b]$, and $\varepsilon\in[0,\varepsilon_0]$.
\end{lemma}

Denote by $t(z)$ and $f(z)$ the multiplication operators by $t(\bs\theta,z)$ and $f(\bs\theta,z)$
in $\LL^2(\TT^d)$, respectively.
By definition of $t(\bs\theta,z)$ for any $\\xi\in\LL^2(\TT^d)$
\begin{multline}
        \label{eq-comm}
 e^{t(z)} e^{f_0(z)} U e^{-t(z)}\xi(\bs\theta)\\
 =
  e^{t(\bs\theta,z)} e^{f_0(\bs\theta,z)} \exp\big({}-t(z,e^{-2\pi i \omega_1}\theta_1,\dots,e^{-2\pi i \omega_d}\theta_d)\big) U \varphi(\bs\theta)\\
  =\exp\big(t(\bs\theta,z)-Ut(\bs\theta,z)+f_0(\bs\theta,z)\big) U\varphi(\bs\theta,z)=e^{f(z)}U\varphi(\bs\theta)=C(z)U\varphi(\bs\theta). 
\end{multline}
Therefore, one can rewrite Eq.~\eqref{eq-MA} as
\begin{equation}
     \label{eq-MA2}
N(z)-B=D(z)^{-1} e^{t(z)}\big(1-e^{f_0(z)}\chi U\big)e^{-t(z)}\big(1+\chi U\big)^{-1}.
\end{equation}

\begin{prop}\label{prop-ev} The set of the eigenvalues of $H$ in $[a,b]$
is dense and 
coincides with the set of solutions $\lambda$ to
\begin{equation}
        \label{eq-mlx}
\sigma(\lambda)=\pi(\bs\omega\vc m +\varphi\big)\mod \pi,\quad \vc m\in\ZZ^d.
\end{equation}
Each of these eigenvalues
is simple, and for any fixed $\vc m\in\ZZ^d$ Eq.~\eqref{eq-mlx} has
at most one solution $\lambda(\vc m)$, and $\lambda(\vc m)\ne\lambda(\vc m')$
for $\vc m\ne \vc m'$.
\end{prop}

\begin{proof}
As follows from proposition~\ref{krein} and the resolvent formula \eqref{eq-nz}, the eigenvalues $\lambda$ of $H$
outside $\spec H_0$ are determined by the condition $\ker\big(N(\lambda)-B\big)\ne 0$,
an their multiplicity coincides with the dimension of the corresponding kernels.
Eq.~\eqref{eq-comm} shows that the condition $(N(\lambda)-B)u=0$
is equivalent to
$\big(1-e^{f_0(\lambda)}\chi U\big)e^{-t(\lambda)}\big(1+\chi U\big)^{-1}u=0$
or, denoting $v:=e^{-t(\lambda)}\big(1+\chi U\big)^{-1}u$,
$(1-e^{f_0(\lambda)}\chi U\big)v=0$, which can be rewritten as
\begin{equation}
    \label{eq-uv}
\chi Uv=e^{-f_0(\lambda)}v,\quad v\ne 0.
\end{equation}
As $\chi U$ has the simple eigenvalues $e^{-2\pi i (\bs\omega\vc m+\varphi)}$, $\vc m\in\ZZ^d$,
and the corresponding eigenvectors form a basis, Eq.~\eqref{eq-uv}
implies \eqref{eq-mlx} if one takes into account the identity $f_0(\lambda)=2i\sigma(\lambda)$
proved in lemma~\ref{prop-1}. The rest follows
from the monotonicity of $\sigma$, the inclusion
$\ran \sigma\subset (-\pi/2,\pi/2)$, and the arithmetic properties
\eqref{eq-vf1} and \eqref{eq-vf2}.
\end{proof}

As $[a,b]$ was an arbitrary interval from $\RR\setminus\spec H_0$,
one has an immediate corollary
\begin{prop}\label{spec3}
The pure point spectrum of $H$ is dense in $\RR\setminus \spec H_0$.
\end{prop}
Now it remains to show that the spectrum of $H$ in the interval considered is pure point.

Take some $\alpha>0$. For any $\delta>0$ we denote
\begin{equation*}
\SSS^1_\delta=\bigcup_{m\in\ZZ^d} \Big\{
\xi\in \SSS^1: |\Arg \xi-\Arg e^{2\pi i\bs\omega\vc m}|\le \delta\big(1+|\vc m|\big)^{-d-\alpha}
\Big\},\quad
\widetilde\SSS^1_\delta:=\SSS^1\setminus \SSS^1_\delta.
\end{equation*}
Clearly, there holds
\begin{equation}
      \label{eq-ss2}
|1-\xi e^{-2\pi i\bs \omega \vc m}|\ge 2\pi^{-1}\delta \big(1+|\vc m|\big)^{-d-\alpha},\quad \xi\in \widetilde\SSS^1_\delta,
\quad m\in\ZZ^d.
\end{equation}

Let $\Delta\subset[a,b]$ be an interval whose ends are not eigenvalues of $H$.
Consider the mapping $h:\lambda\mapsto \chi e^{f_0(\lambda)}$.
By~lemma~\ref{prop-1}, $h$ is a diffeomorphism between $\Delta$ and $h(\Delta)$.
By proposition~\ref{prop-ev} one has $h\big(\lambda(\vc m)\big)=e^{2\pi i\bs\omega\vc m}$.
Take an arbitrary $\delta>0$ and denote
\[
\Delta_\delta:=\Delta\cap h^{-1}(\SSS^1_\delta),\quad
\widetilde\Delta_\delta:=\Delta\cap h^{-1}(\widetilde\SSS^1_\delta)\equiv\Delta\setminus\Delta_\delta.
\]
Clearly, $\Delta_\delta$ is a countable union of intervals, and
the limit set $\bigcap_{\delta>0}\Delta_\delta$
coincides with the set of all the eigenvalues $\bigcup_m\{\lambda(m)\}$.

\begin{lemma}\label{lem-fin}
There exists $\varepsilon_0>0$ such that for any $\delta>0$ and any $\vc n\in\ZZ^d$ there exists $C>0$ 
such that
\begin{equation}
  \label{eq-mla}
\big\|\big(N(\lambda+i\varepsilon)-B\big)^{-1}\delta_{\vc n}\big\|\le C
\end{equation}
for all $\lambda\in\widetilde\Delta_\delta$,
and $\varepsilon\in(0,\varepsilon_0)$.
\end{lemma}

\begin{proof}
Rewrite Eq.~\eqref{eq-MA2} in the form
\[
\big(N(z)-B\big)^{-1}=\big(1+\chi U\big)e^{t(z)}\big(1-e^{f_0(z)}\chi U\big)^{-1}e^{-t(z)} D(z).
\]
Note that the Fourier transform of $\delta_{\vc n}$ is the function $\bs\theta\mapsto \bs\theta^{\vc n}$.
Denote $\Psi(z;\bs\theta):=e^{-t(\bs \theta,z)}B(\bs\theta,z)\bs \theta^{\vc n}$.
Due to the analyticity one can estimate uniformly in $Z$:
\[
|\psi_z(\vc m)|\le C' e^{-\rho|\vc m|},\quad C',\rho>0, \quad \psi_z:=F_d^{-1}\Psi,\quad
\|(1+\chi U)e^{t(z)}\|\le C'.
\]
Therefore, \eqref{eq-mla} follows from the inequality
\begin{equation}
 \label{eq-fini}
\big\|
(1-e^{f_0(\lambda+i\varepsilon)}\chi U)^{-1}\Psi
\big\|\le C.
\end{equation}
Assume that $\varepsilon_0$ satisfies the conditions of lemma~\ref{lem-eps}, then
uniformly for $\lambda\in\Delta$ and $\varepsilon\in(0,\varepsilon_0)$ one has
\begin{multline*}
\big|\big(F_d^{-1}(1-e^{f_0(\lambda+i\varepsilon)}\chi U)^{-1}\Psi\big)(\vc m)\big|
=\big|
(1-e^{f_0(\lambda+i\varepsilon)}\chi e^{2\pi i \bs\omega \vc m})^{-1}\psi_{\lambda+i\varepsilon}(\vc m)
\big|\\
\le 2
\big|
(1-e^{f_0(\lambda)}\chi e^{2\pi i \bs\omega \vc m})^{-1}\big| \cdot|\psi_{\lambda+i\varepsilon}(\vc m)|.
\end{multline*}
As in our case $h(\lambda)\equiv\chi e^{f_0(\lambda)}\in \widetilde\SSS^1_\delta$, due to \eqref{eq-ss2}
we have
\[
\big|
(1-e^{f_0(\lambda)}\chi e^{-2\pi i \bs \omega \vc m})^{-1}\big|\le \dfrac{\pi}{2\delta}\,\big(1+|\vc m|\big)^{d+\alpha}.
\]
Finally,
\begin{multline*}
\big\|
(1-e^{f_0(\lambda+i\varepsilon)}\chi U)^{-1}\Psi
\big\|^2=\sum_{\vc m\in\ZZ^d} \big|\big(F_d^{-1}(1-e^{f_0(\lambda+i\varepsilon)}\chi U)^{-1}\Psi\big)(\vc m)\big|^2\\
\le \Big(\dfrac{\pi C'}{\delta}\Big)^2\,\sum_{\vc m\in\ZZ^d}\big(1+|\vc m|\big)^{2(d+\alpha)}e^{-2\rho|\vc m|}<\infty,
\end{multline*}
and \eqref{eq-fini} is proved.
\end{proof}

Now we are able to estimate the spectral projections corresponding to $H$.

\begin{lemma}\label{lem-fff} For any $f\in\HH$ and any $\delta>0$ one has
\begin{equation}
  \label{eq-rrr}
\lim_{\varepsilon\to 0+} \varepsilon \int_{\widetilde\Delta_\delta} \|(H-\lambda-i\varepsilon)^{-1}f\|^2d\lambda=0.
\end{equation}
\end{lemma}

\begin{proof}
Here we are going to use proposition~\ref{krein}. First note that due to $\widetilde \Delta_\delta\subset\RR\setminus\spec H_0$ one has
\begin{equation}
     \label{eq-h0}
\lim_{\varepsilon\to 0} \varepsilon \int_{\widetilde\Delta_\delta} \|(H_0-\lambda-i\varepsilon)^{-1}f\|^2d\lambda=0
\text{ for any } f\in\HH.
\end{equation}
Similar to \eqref{eq-h01} let us consider the decomposition
\[
\HH=\HH_0+\HH_1,\quad
\HH_0:=\Big(\bigcup_{\Im z \ne 0} \nu(z)\big(\ell^2(\ZZ^d)\big)\Big)^\perp,
\quad
\HH_1:=\HH_0^\perp;
\]
As previously,
by~\eqref{eq-nz}, for any $f\in\HH_0$ and any $z$ with $\Im z\ne 0$ one has $\nu^*(z)f=0$.
Hence, by~\eqref{eq-krein}, there holds $(H-z)^{-1}f=(H_0-z)^{-1}à$, and \eqref{eq-h0}
implies \eqref{eq-rrr} for $f\in\HH_0$.

Now it is sufficient to show \eqref{eq-h0} for vectors $f=\nu(\zeta)h$
for $h=\big(N(\zeta)-B\big)^{-1}\delta_{\vc m}$, $\vc m\in\ZZ^d$,
$\Im\zeta\ne0$. The operators $\big(N(\zeta)-B\big)^{-1}$ have dense range
(coinciding with $\dom B$),
hence the linear hull of such vectors $f$ is dense in $\HH_1$.
By elementary calculations (see e.g. section~3 in~\cite{BGP}) one rewrites Eq.~\eqref{eq-krein} as
\begin{equation}
   \label{eq-res}
(H-\lambda-i\varepsilon)^{-1}f=
\dfrac{1}{\zeta-\lambda-i\varepsilon}\,\Big(
f-\nu(\lambda+i\varepsilon) \big(N(\lambda+i\varepsilon)-B\big)^{-1}\delta_{\vc m}\Big).
\end{equation}
Due to lemma~\ref{lem-fin} we have $\big\|\big(N(\lambda+i\varepsilon)-B\big)^{-1}\delta_{\vc m}\big\|\le C$ with some $C>0$,
for all $\lambda\in\widetilde\Delta_\delta$ and sufficiently small $\varepsilon$,
and \eqref{eq-res} implies
\[
\|(H-\lambda-i\varepsilon)^{-1}f\|\le \dfrac{\|f\|+C\|\nu(\lambda+i\varepsilon)\|}{|\zeta-\lambda-i\varepsilon|},
\]
and due to the analyticity of $\gamma$, one can estimate $\|(H-\lambda-i\varepsilon)^{-1}f\|\le C'$
with some $C'>0$ for all $\lambda\in\widetilde\Delta_\delta$ and sufficiently small $\varepsilon$.
This obviously implies \eqref{eq-rrr}.  
\end{proof}

\begin{prop}\label{spec4}
The spectrum of $H$ outside $\spec L$ is pure point.
\end{prop}

\begin{proof}We are going to show that for any $f\in\HH$ and any interval $\Delta\subset\RR\setminus\spec H_0$
the spectral measure $\mu_f$ associated with $H$ and $f$ satisfies
$\mu_f(\Delta)=\mu_f\big(\Delta\cap\bigcup_m\{\lambda(m)\}\big)$; this proves that all
the spectral measures are pure point.

By the Stone formula, for any set $X$ which is a countable union of intervals
whose ends are not eigenvalues of $H$ one has
\[
\mu_f(X)=\lim_{\varepsilon\to0+}\dfrac{\varepsilon}{\pi}\int_X\|(H-\lambda-i\varepsilon)f\|^2d\lambda.
\]
Using lemma~\ref{lem-fff}, for any $\delta>0$ we estimate
\begin{multline*}
\mu_f(\Delta)=\lim_{\varepsilon\to0+}\dfrac{\varepsilon}{\pi}\int_\Delta\|(H-\lambda-i\varepsilon)f\|^2d\lambda\\
=\lim_{\varepsilon\to0+}\dfrac{\varepsilon}{\pi}\int_{\Delta_\delta}\|(H-\lambda-i\varepsilon)f\|^2d\lambda
+\lim_{\varepsilon\to0+}\dfrac{\varepsilon}{\pi}\int_{\widetilde\Delta_\delta}\|(H-\lambda-i\varepsilon)f\|^2d\lambda\\
=\lim_{\varepsilon\to0+}\dfrac{\varepsilon}{\pi}\int_{\Delta_\delta}\|(H-\lambda-i\varepsilon)f\|^2d\lambda=\mu_f(\Delta_\delta).
\end{multline*}
As $\delta$ is arbitrary and $\bigcap_{\delta>0}\Delta_\delta=\bigcup_m\{\lambda(m)\}$, the theorem is proved.
\end{proof}

\section*{Acknowledgments}

The work was supported by the Marie Curie Intra-European Fellowship
PIEF-GA-2008-219641 during the stay at the University Paris Nord in July-Septembre 2008.

\end{document}